\documentclass{article}
\usepackage[ruled]{algorithm2e}

\SetAlFnt{\small}
\SetAlCapFnt{\small}
\SetAlCapNameFnt{\small}
\SetAlCapHSkip{0pt}
\IncMargin{-\parindent}

\usepackage{color}
\usepackage{subfigure}
\usepackage{listings}
\usepackage{graphicx}
\usepackage[standard]{ntheorem}

\begin{document}

\title{Using an Epidemiological Approach to Maximize Data Survival in the Internet of Things}
 \author{Abdallah Makhoul, Christophe Guyeux, Mourad Hakem,\\ and Jacques M. Bahi}
 
\maketitle

\begin{abstract}
The internet of things (IoT) has gained worldwide attention in recent years. It transforms the everyday objects that surround us into proactive actors of the Internet, generating and consuming information. An important issue related to the appearance of such large-scale self-coordinating IoT is the reliability and the collaboration between the objects in the presence of environmental hazards. High failure rates lead to significant loss of data. Therefore, data survivability is a main challenge of the IoT. In this paper, we have developed a compartmental e-Epidemic SIR (Susceptible-Infectious-Recovered) model to save the data in the network and let it survive after attacks. Furthermore, our model takes into account the dynamic topology of the network where natural death (crashing nodes) and birth are defined and analyzed. Theoretical methods
and simulations are employed to solve and simulate the system of equations developed and to analyze the model. 
\end{abstract}

\section{Introduction}

The Internet of Things (IoT) or Internet of objects has gained worldwide attention in recent years, particularly with the proliferation of new communication technologies and connected devices. The main idea behind the IoT is to bridge the gap between the physical world of humans and the virtual world of electronics via smart objects. These smart objects allow the interactions between humans and their environment by providing, processing and delivering any sort of information or command. This concerns a large variety of IoT applications, which contribute to our everyday life. They cover a wide range from traditional equipment to general household objects, which help to make people's lives easier~\cite{17,18}. Sensors and actuators will be integrated in buildings, vehicles, and common environment and can tell us about them, their state or their surroundings.

The IoT built from smart things or objects needs to address challenges related to system architecture, design and development, integrated management, business models and human involvement. These challenges will be addressed by: 
\begin{itemize}
\item energy issues in all their phases, from harvesting to conservation and usage, are essential in the development of the IoT;
\item scalability, IoT applications require a large number of devices where it is difficult to implement due to restrictions on time, memory, processing, and energy constraints; 
\item standardization and interoperability of data processing awareness are highly needed; 
\item big data management, it is expensive to transmit huge volumes of raw data in the complex and heterogeneous network. Therefore, IoT needs data compression and data fusion to reduce the data volume; 
\item interaction between hardware, software, algorithms as well as the development of smart interfaces among things; 
\item security, hackers, malicious software and virus in the communication process might disturb data and information integrity.  
\end{itemize}
With the development of the IoT technology, we must develop new techniques and concepts to improve the existing security and privacy in order to adapt to new technological and societal challenges~\cite{19}.

With this increasing use of technology and its social life applications such as intelligent transportation, smart cities, smart homes, etc~\cite{19}, the use of the Internet increases, constantly offering new functionalities and facilities. Although these applications can be extremely useful, they must ensure personal privacy else private information may be leaked at any time. Furthermore, the IoT using wireless sensor networks proved crucial in disaster and rescue missions such as natural disasters (e.g. earthquakes), life-threatening mining accidents, monitoring of critical infrastructures, etc. A major obstacle that delays the appearance of such large-scale self-coordinating IoT is the reliability of the objects/sensors in the presence of environmental hazards. High failure rates lead to a significant loss of data. Therefore, data survivability is a main challenge of the IoT. For instance, in urban disaster areas, the collected data can identify hazards and save lives whereas nodes failures or attacks lead to data loss. Another example is about critical infrastructure monitoring where data can be lost after a portion of the critical infrastructure suffers a disaster. Thus, data survivability and availability is particularly important in the IoT and cannot be ignored. Therefore, collaboration and transmitting crucial information between nodes are essential to maximize the amount of monitoring-related data that can survive.

For that purpose, in this paper we study and develop a new epidemic-domain inspired approach to model the information survivability in the IoT. Somehow, the propagation of the information in a network of things could be compared with a disease transmitted by vectors when dealing with public health. In~\cite{20} the authors discussed the spreading nature of biological viruses, leading to infectious diseases in human populations through several epidemic models. 

The propagation of the information throughout the IoT can be studied by using epidemiological models for disease propagation. The model we present here is based on the SIR (Susceptible - Infected - Recovered) model. A node/thing is susceptible to a data item when it is online and functioning normally; it can receive the information that must survive. Based on a classical epidemic model, various dynamic models for malicious attacks propagation were proposed~\cite{21,22,23,24,25,26,27,28,29,30,31}. The majority of these models were studied for powerful computer networks and based on the fully-connected assumption of the network which is not the case of the IoT based on wireless sensor networks with heavy resources and a very dynamical topology. In this paper, we introduce a thorough analysis of the conditions that can assure data survivability in the internet of things. We study a new SIR model that considers dynamic topologies and nodes energy constraints. Our novelty in this paper is that we study arbitrary dynamic network topologies instead of static networks. We establish a new information propagation model which incorporates the effects of the dynamic IoT topology and its heavy resources. The dynamics of this model are studied, specifically, the level of the attacks and the disparition of nodes. Some numerical examples are given to support this result. 

The remainder of the paper is organized as follows: Section~\ref{PW} briefly reviews the related work.
The SIR model for Data Survivability in the IoT is presented in Section~\ref{SIR}. Sections \ref{USIR} and \ref{VSIR} detail the proposed epidemic schemes in a comparative manner and give variations of SIR applied to the IoT. In these sections theoretical and numerical results will be presented. 
Finally, Section~\ref{CONC} concludes this research work.

\section{Related work}
\label{PW}
In the literature, we can find several mathematical models which illustrate the dynamical behavior of the transmission of biological diseases and/or computer viruses. Based on the Kermack and McKendrick SIR classical epidemic model~\cite{Kermack27,32}, dynamical models for malicious objects propagation were proposed. Due to the numerous similarities between biological viruses and computer viruses, several approaches and models are proposed to study the spreading and attacking behavior of computer viruses in different phenomena, e.g. virus propagation~\cite{33,34,35}, e-mail propagation schemes~\cite{36}, virus immunization~\cite{37,38}, quarantine~\cite{39,40}, vaccination~\cite{41}, etc.
The authors in~\cite{42} propose an improved SEI (susceptible-exposed-infected) model to simulate virus propagation. \cite{43} propose an SEIS-V epidemic model with vertical transmission using vaccination (that is, run of anti-virus software time and again with full efficiency) so that a temporary recovery from the infection of worms can be obtained.

More recently, epidemiological models have been used not only to transmit viruses in computer network but also to ensure the security in wireless sensor networks~\cite{44,45,DiPietro11,DiPietro13}. The authors in~\cite{44,45} studied the robustness of filtering on nonlinearities in packet losses and sensors. Unattended Wireless Sensor Networks (UWSNs), have been introduced by Di Pietro \emph{et al.} in~\cite{DiPietro08}, where adversaries can compromise some sensor nodes and selectively destroy data. In such networks, nodes collect data from the area under consideration, and then they try to upload all the stored data when the sink comes around and the main challenge is data survivability. The epidemiology community has developed the so-called SIR and SIS models~\cite{DiPietro11,DiPietro13} of infection. The SIS model (Susceptible - Infected - Susceptible) is suitable for, e.g., the common flu, where nodes may be infected, healed (and susceptible), and infected again. The SIR model (Susceptible- Infected - Recovered) is for example suitable for mumps, where a node, after being infected, becomes recovered (with life-time immunity). SIS, SIR, and SIRS models have been investigated by authors of these research works, in order to derive the parameters that can ensure information to survive. In these articles, the $S(t)$ compartment is constituted by sensors that do not possess the datum at time $t$, while $I(t)$ is the compartment of sensors that possesses it. Finally, the $R(t)$ compartment is constituted by sensors that have been compromised by the attacker.  

To the best of our knowledge none of the previous work have studied epidemic models for data survivability in the IoT. Furthermore, existing approaches have not taken into account the energy consumption constraints of the nodes. As in the IoT, usually the nodes' energy is provided by a battery that can be emptied due to data acquisition, transmission, or simply the functioning cost of keeping nodes alive. On the other hand, the topology of the networks they consider is static, the network's lifetime is unbounded, and nodes cannot die due to empty batteries. Our intention in this paper is to provide a new epidemic model dedicated to the IoT, by taking into account these issues.

\section{Formulation of SIR model for IoT}
\label{SIR}
\subsection{Introducing the Kermack \& McKendrick model}
 
In IoT, the global network can be divided in three
compartments, namely the nodes/things $S$ susceptible to receive the datum of
interest (intrusion detection, etc.), the ones that currently store it (Informed) $I$, 
and the recovered nodes $R$ that have been 
compromised by the attacker: their stored datum has been recovered.

Suppose now that between $S$ and $I$, the transmission rate is $b I$, where
$b$ is the contact rate, which is the probability of transferring the information in a
contact between a susceptible node and another having the datum. Indeed, as proven
by Di Pietro \emph{et al.}, such a situation occurs when the network is composed
by $N$ nodes, and if each node forwards the datum with
probability $\frac{\alpha}{N}$~\cite{DiPietro11,DiPietro13} ($\alpha$ is the transition rate). 

Suppose, additionally, that the rate to pass between $I$ and $R$, is $c$:
the attacker is able to individuate the nodes containing the target
information, and to destroy each of them with this probability $c$.
Notice that, if the duration of the information survivability
is $D$, then $c = \frac{1}{D}$, as a node experiences one recovery
in $D$ units of time.

\begin{figure}[ht]
\centering
\includegraphics{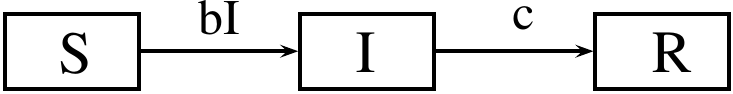}
\caption{SIR model}
\label{SIRmodel}
\end{figure}

Under such hypotheses and as stated in~\cite{DiPietro11,DiPietro13}, 
the nodes population follows the so-called SIR model 
of Kermack \& McKendrick~\cite{Kermack27} depicted in Figure~\ref{SIRmodel}.
Notice that the total sensors population is equal to $N=S+I+R=S_0+I_0+R_0$, 
which is a constant: the number of connected nodes does not evolve. 
In particular, only two of the three populations of nodes have to be studied.

\subsection{First Theoretical Study}
The time dependant SIR model can be expressed by the following set of ordinary non-linear differential equations : 
\begin{equation}
\label{modelSir1}
\left\{
\begin{array}{l}
\frac{dS}{dt} = - b I S\\\\
\frac{dI}{dt} = b I S - c I\\\\
\frac{dR}{dt} = c I .\\
\end{array}
\right.
\end{equation}

\begin{figure}[ht]
\centering
\includegraphics[scale=0.6]{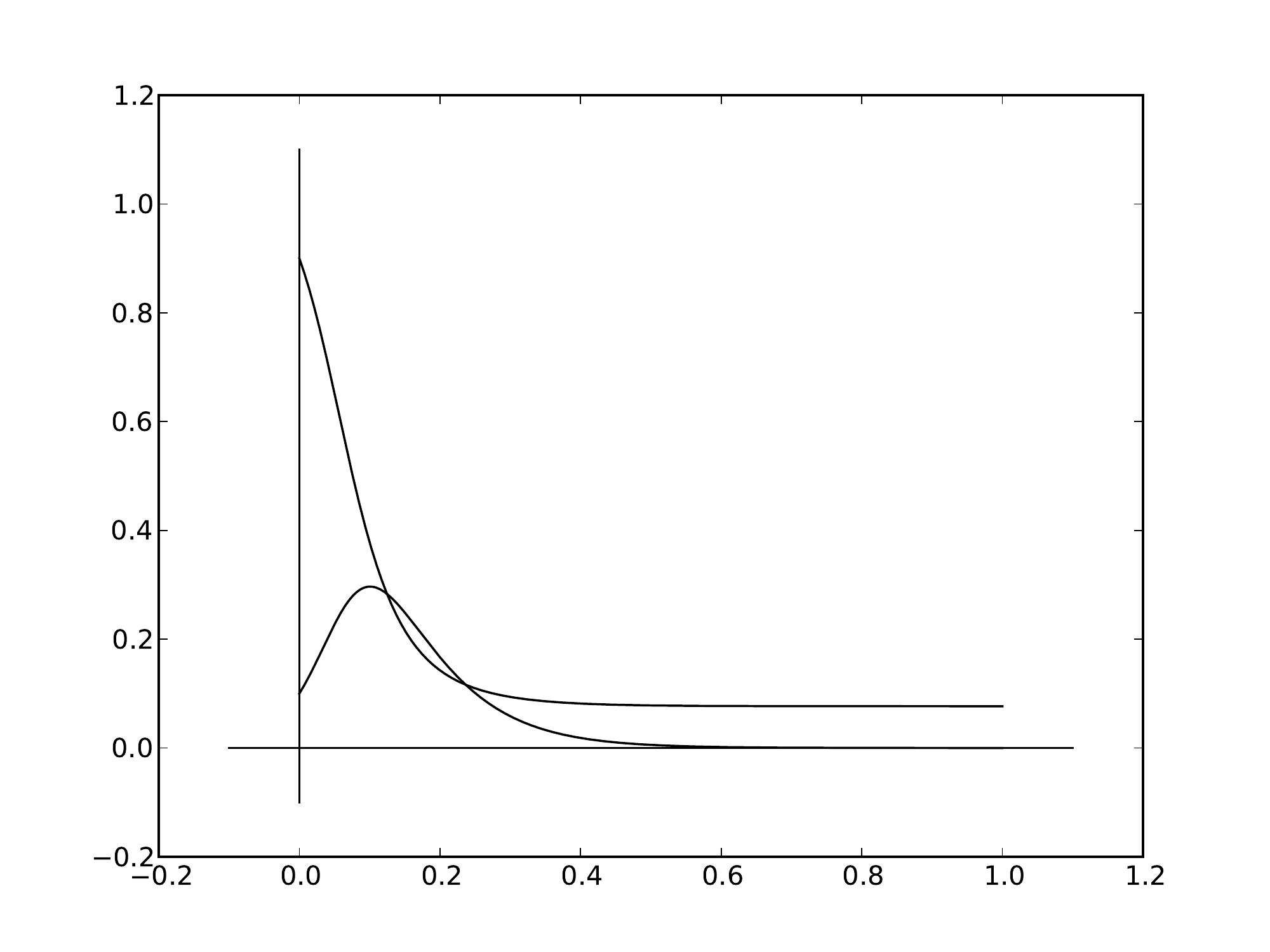}
\caption{Evolution of the fractions $S$ and $I$ of susceptible and having the datum nodes with b = 0.4, c = 0.15, s(0) = 0.9, and i(0) = 0.1 (SIR model).}
\label{imggg2}
\end{figure}

We suppose that each informed node (which has received the datum) communicates with $k$ nodes per unit time, where $k$ is independent of the size of the network. Thus, it communicates with $\frac{kS}{N}$ susceptible nodes $S$. If a fraction $\tau$ of these nodes receives the information, then each informed node $I$ communicates the information to $\frac{\tau kS}{N}$ new susceptible nodes $S$ per unit time, and $b=\frac{\beta}{N}$, where $\beta = k \tau$. We called $\tau$ the transmissibility of the information.

In the equation system~\ref{modelSir1}, we can see that the right member of the first line is negative, and the right member of the third line is positive. Therefore, we deduce that $\frac{dS}{dt} \leq 0$ and $\frac{dR}{dt} \geq 0$, then, assuming that the quantities $S$ and $R$ are positive we obtain:

\medskip
\begin{itemize}
\item $0\leq S(t) \leq S(0) \leq N$,
\item $0\leq R(0) \leq R(t) \leq N$.
\end{itemize}
\medskip

We can observe that $S$ is a decreasing function and bounded from below, and it converges to a limit denoted by $S(\infty)$). Similarly, $R$ converges to $R(\infty)$, and then $I=N-R-S$ converges also, to $I(\infty) = N - S(\infty) - R(\infty)$.

We can deduce from the above that the number of the informed nodes approaches $0$ ($I(\infty) = 0$). Indeed, if this were not the case and since $\frac{dR}{dt} = c I$, we can deduce that for $t$ large enough, $\frac{dR}{dt} > c \frac{I(\infty)}{2}>0$ and $ R(\infty) = \infty$, which is absurd. 

\section{The SIR Model study}
\label{USIR}
\subsection{First Theoretical Results}

Let $R_e = \frac{S(0) \beta}{Nc}$ denote the number of the effective information reproductions in the network, and $R_0 = \frac{\beta}{c}$ be the basic reproduction number. 
If, at time $t=0$ the entire network (all the connected nodes) are susceptible and only one node informed (which means that $S(0)=N-1$ et $I(0)=1$), and if the network is large, then $R_e = \frac{(N-1) \beta}{Nc}$ is approximately equal to $R_0$, which we are going to assume in the reminder of the paper.  

\medskip

Under this assumption of large network, we prove that $R_e$ is the limit value that determines whether the information will propagate within the network, or whether it will quickly disappear. 

\begin{proposition}
\label{prop1}
If $R_e \leq 1$, then $I(t)$ decreases to $0$ as $t \rightarrow \infty$. Else, $I(t)$ grows to a maximum, then decreases to $0$ (epidemic transmission of the information).
\end{proposition}

As shown in figure~\ref{imggg1}, the information propagates in the networks as an epidemic-like if $R_0>1$, which means if $\beta > c$. 

\begin{proof}

From the second equation of the system~\ref{modelSir1}, we can deduce that 

$$\frac{dI}{dt} = (bS-c)I \leq (bS(0)-c)I = c(R_e-1)I\leq 0$$ 

for $R_e\leq 1$. As $I(\infty)=0$, the first result of the proposition is then obtained.

Similarly, the second equation of the system~\ref{modelSir1} implies

$$\frac{dI}{dt}(0) =  c(R_e-1)I(0)>0$$

for $R_e>1$. Thus, the function $I$ increases for $t$ approaches $0$. This equation implies also that there is no non-zero constant value of $ I $ satisfying Equation 2. These remarks and the fact that $I(\infty)=0$ complete the proof. 
\end{proof}  

\subsection{Finding the maximum number of informed nodes $I$}

Dividing the first two equations of the system of the SIR model, gives:

$$\frac{dS}{dI} = \frac{-b SI}{b SI -c I}.$$

This differential equation independent variables can be rewritten as follows:

$$\int \frac{bS-c}{bS} dS = - \int dI .$$

We find then $-I-S+\frac{c}{b} \ln{S} = k$, where k is a constant and thus $\forall t \geq 0$,

\begin{equation}
\label{eqIS}
I(t)+S(t)-\frac{c}{b} \ln{S(t)} = I(0)+S(0)-\frac{c}{b} \ln{S(0)} .
\end{equation}

The maximum number of the informed nodes $I_{max}$ is reached when $\frac{dI}{dt}=0$. According to the system~\ref{modelSir1}, this occurs when $S=\frac{c}{b}$. 
Rewriting the equation~\ref{eqIS}, we find:

$$I_{max} = I(0)+S(0)-\frac{c}{b} \ln{S(0)} - \frac{c}{b} \left(1-\ln{\frac{c}{b} }\right).$$

Particularly, assuming that $I(0)=1$ and $S(0)=N-1$, we find 

$$I_{max} = N - \frac{N}{R_0}\left( 1+ \ln{R_0}\right).$$

Moreover, from equation~\ref{eqIS} one can deduce that, in the plane $S-I$, the solutions $(S(t),I(t))$ are in the contour lines of the function $F(S,I)=S+I-\frac{c}{b} \ln {S}$ as shown in figure~\ref{imggg1}   

\begin{figure}[ht]
\centering
\includegraphics[scale=0.6]{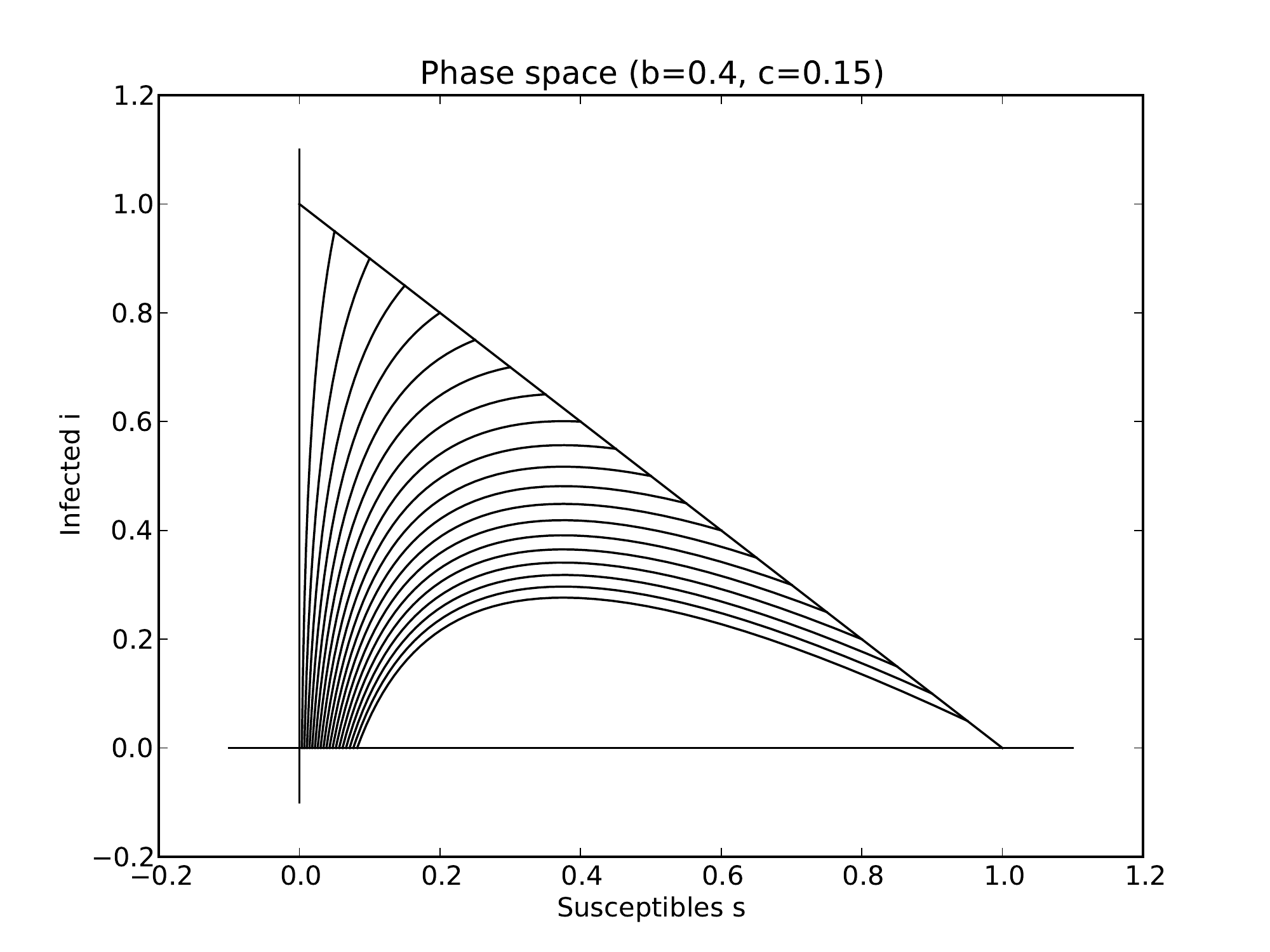}
\caption{Phase space $(S,I)$ with $b=0,4, c=0,15$ (SIR Model).}
\label{imggg1}
\end{figure}

\subsection{Stopping the information transmission}

In the context of data survivability for IoT, a question could be asked about the stopping of the information transmission because there are no more nodes susceptible to receive it (which means $S(\infty)=0$). The next proposition proves that this case cannot happen.

\begin{proposition}

The minimum number of nodes capable of receiving information, satisfies the following inequality: 

 $$S(\infty) \geq S(0) \exp(-R_0).$$

In particular, this limit is strictly positive.
\end{proposition}

\begin{proof}

We proceed as in the previous proposition, by dividing now the first to the third equation in the system~\ref{modelSir1} as follows:

$$\frac{dS}{dR} = \frac{-bIS}{cI} = \frac{-bS}{c}.$$

We find again, a differential equation with independent variables which can be rewritten as:

$$\int \frac{dS}{S} = \int \frac{-b}{c} dR,$$

Which implies $S(t)= S(0) \exp\left(-b\frac{R(t)-R(0)}{c}\right)$.

Since $0 \leq R(t)-R(0) \leq N$, then $S(t)\geq S(0) \exp \left( \frac{-bN}{c} \right)$, and thus $$S(\infty) \geq S(0) \exp \left( -\frac{bN}{c} \right) = S(0) \exp (-R_0) >0.$$

Noting that $S$ is actually an integer, then for the high value of $R_0$, we find that
$S(\infty) = 0$: in this extreme case, all nodes will receive the information at a given time. 
\end{proof}

In the case when an epidemic propagation of the information occurs in the network, the number of the nodes susceptible decreases. 
Thus the rate of appearance of new nodes informed also decreases, and at a time $t$ it may happen that $S(t)$ is less than $\frac{c}{b}$: the rate of
the nodes having lost the information exceeds the rate of the informed nodes, this explains why $I(t)$ starts decreasing. The transmission of information
will then stop in the network due to a lack of informed nodes (and not because of the number of susceptible nodes). 

\section{Variations of SIR applied to the IoT}
\label{VSIR}
\subsection{Another understandings for the Recovered compartment}

In the previous section, the $R$ compartment was constituted by 
nodes that have been compromised by the attacker, which will
be referred in what follows as situation 1.
It is possible to attribute at least two other understandings to
this compartment, for nodes in the IoT based on a wireless sensor network 
whose lifetime is dependent on energy consumption and in absence of attacks.

This compartment can be constituted by dead nodes, when considering 
that the sole action on the energy is the information transmission, and
that the unique way to die for a node is to have transmitted too many data. In other
words, in this Situation~2, the nodes send information messages to their neighbors until totally emptying
 their batteries. The user will receive the information when it interrogates the network at time $t$ if $I(t) \neq 0$.

A third situation can be considered without any changes in formalization,
except redefining the meaning of the $R$ compartment. Indeed, it can be
interesting to consider that a node is first susceptible to receive an information message
for a while, then once the message has been received it can be transmitted, before 
finally entering into the third age of its life, the recovered state in
which it will lose its ability to transmit the information. 

However, in many situations of the IoT, the energy consumption and the death of the nodes can not be neglected. This is why a natural extinction rate of nodes for each compartment $S$, $I$ and $R$ will be introduced in the next section.

\subsection{A SIR model for the IoT with a natural death rate}
\label{USIR1}

\begin{figure}[ht]
\centering
\subfigure[Situation 2]{\label{SIR2sit2}
\includegraphics[scale=0.65]{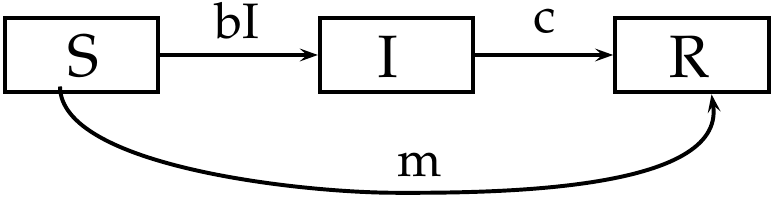}}\hspace{1cm}
\subfigure[Situations 1 and 3]{\label{SIR2sit13}
\includegraphics[scale=0.65]{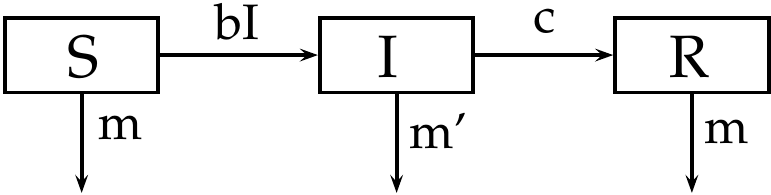}} 
\caption{SIR models with natural death rate}
\label{SIR2modeldeath}
\end{figure}

The previous section considers that all nodes activities are negligible, in terms of 
energy, except the transmission of information in Situations 2 and 3, which is reasonable in a
first approximation. It is however possible to refine the SIR model in these two last situations, in order to 
consider that the nodes' energy decreases too in absence of information transmission.

\begin{figure}[ht]
\centering
\subfigure[Situation 2]{\includegraphics[scale=0.33]{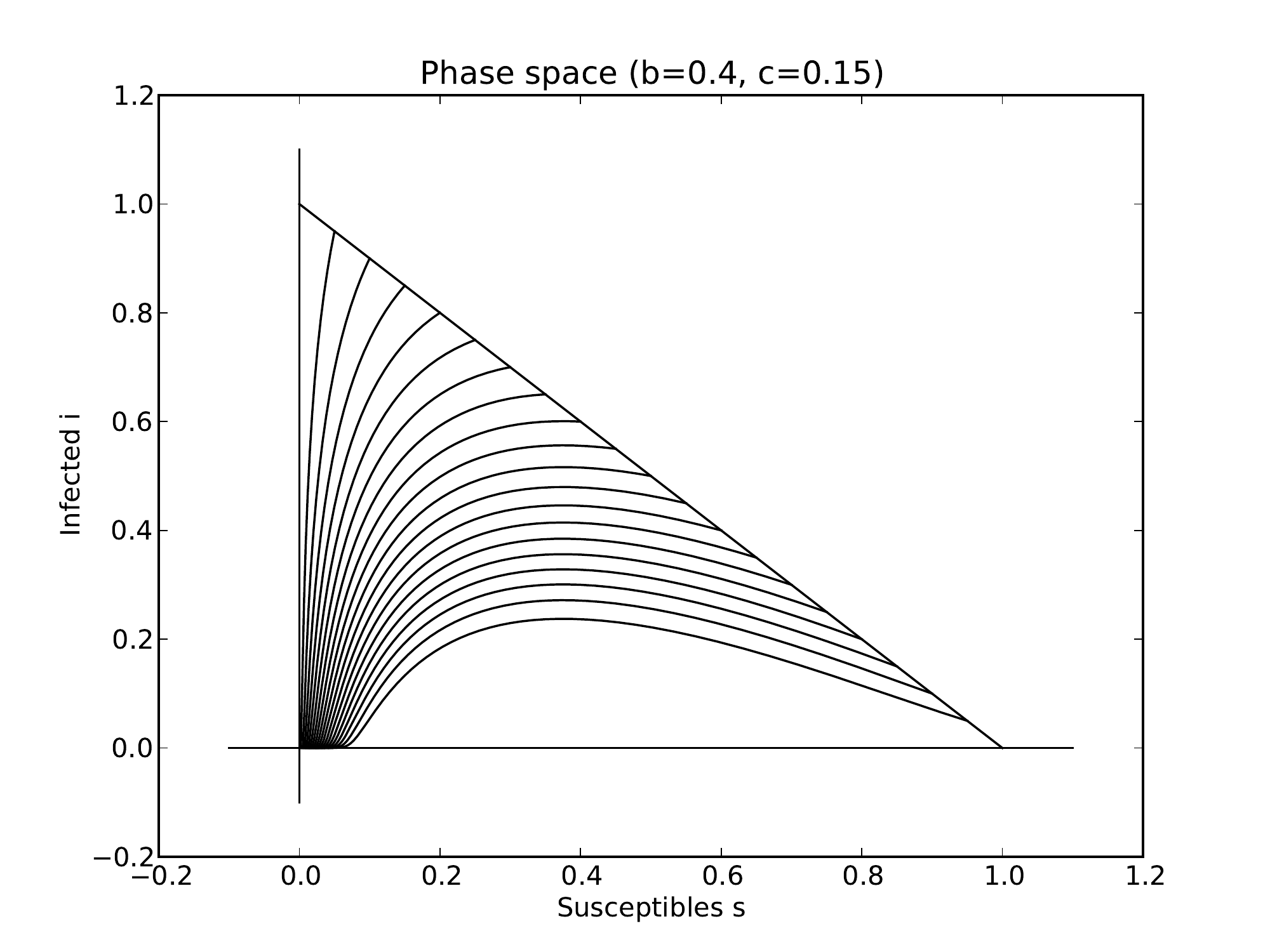}}
\subfigure[Situations 1 and 3]{\includegraphics[scale=0.33]{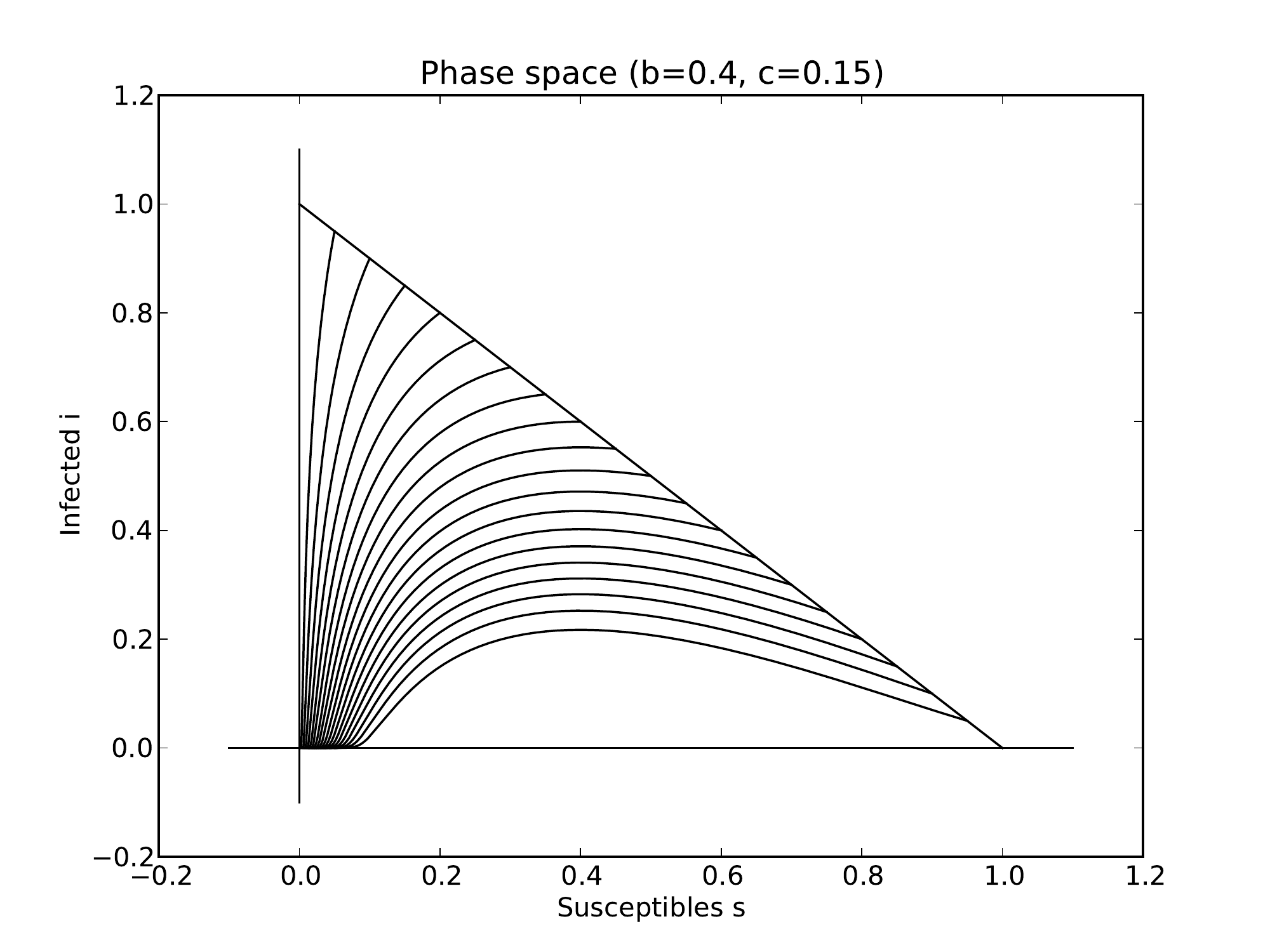}}
\caption{Phase space $(S,I)$ with $b=0.4, c=0.15, m=0.01$, SIR model with natural death rate in Situation 3.}
\label{img1}
\end{figure}

In Situation 2, the $R$ compartment of the SIR model is constituted by 
dead nodes. This compartment is populated by susceptible nodes that have
naturally died (death rate $m$) without having received the datum and by
nodes of the $I$ compartment which die at another rate $c$ supposed
to be greater than $m$, as they have to transfer the datum, an energy-consuming
task. This situation is depicted in Figure~\ref{SIR2sit2}.

\begin{figure}[ht]
\centering
\includegraphics[scale=0.6]{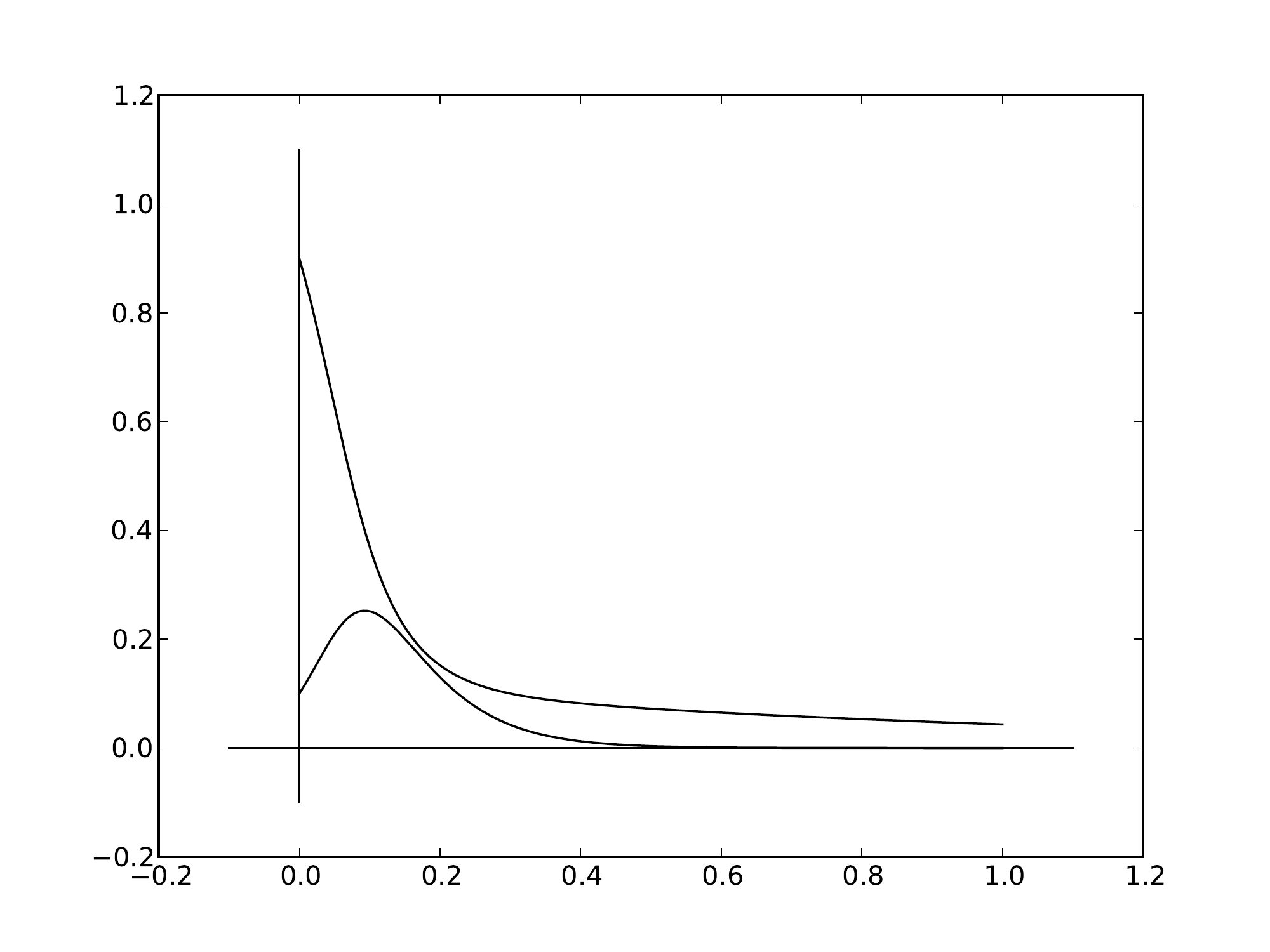}
\caption{Evolution of the fractions $S$ and $I$ of susceptible and having the datum 
nodes with $b=0.4, c=0.15, m=0.01, s(0)=0.9,$ and $i(0)=0.1$, SIR model with natural nodes death rate in Situations 1 and 3.}
\label{img2}
\end{figure}

In the two other situations investigated in this research work, the $R$ 
compartment is constituted by living nodes that do not transmit the
datum anymore, either because they have been corrupted and thus have 
lost it (first situation), or because their batteries are preserved (third one).
This new situation is closed to the SIR model of Figure~\ref{SIRmodel}, except
that a new network is characterized by a death rate for each nodes
compartment (see Figure~\ref{SIR2sit13}). 
Notice that the death rate $m'$ of the $I$ compartment is \emph{a priori} different
from the one of $S$ and $R$ compartments, as it is reasonable to suppose that
the datum transmission implies more energy consumption. 

The SIR model of Equation~\ref{modelSir1} can be adapted as follows for Situation 2:
\begin{equation}
\label{modelSir2}
\left\{
\begin{array}{l}
\frac{dS}{dt} = - b I S -m S\\\\
\frac{dI}{dt} = b I S - c I\\\\
\frac{dR}{dt} = c I+mS,\\
\end{array}
\right.
\end{equation}
while it has the following form in Situations 1 and 3:
\begin{equation}
\label{eqSit3}
\left\{
\begin{array}{l}
\frac{dS}{dt} = - b I S -m S\\\\
\frac{dI}{dt} = b i s - c I - m' I\\\\
\frac{dR}{dt} = c I - mR.\\
\end{array}
\right.
\end{equation}

Numerical simulations are used to show the long-term behavior of the system under such assumptions. The phase space for the three situations
are detailed in figure~\ref{img1} and in figure~\ref{img2}. We show the evolution of susceptible $S$ and having the datum $I$ 
nodes for situations 1 and 3.

To put it in a nutshell, to achieve data survivability in 
the IoT, the birth of connected nodes must be considered, which 
is the subject of the next subsection.

\subsection{Achieving data survivability using connected/birth and disconnected/death rates}

\begin{figure}[!h]
\centering
\includegraphics[scale=0.8]{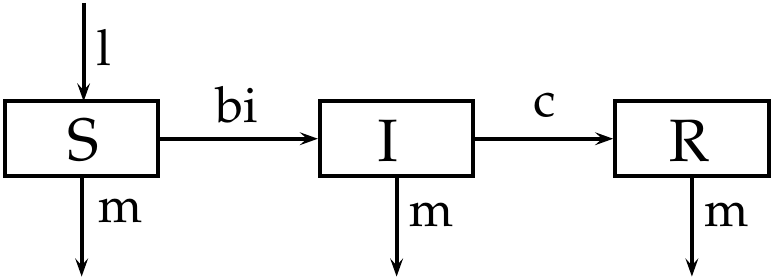}
\caption{SIR model with natural connected/birth and disconnected/death rates}
\label{SIR3modelbirthdeath}
\end{figure}

Considering now a new SIR model that leads to the division of each compartment in two parts, corresponding
respectively to connected and disconnected nodes. Indeed, in the IoT, nodes or things can be scheduled to be connected/awaken or not in order to save energy and extend the nodes lifetime. Therefore, a scheduling process of the nodes can be established where each node periodically decides to be in active/connected or sleep/disconnected mode. We suppose that, initially, a small part of the nodes is connected.
New nodes are then connected periodically during the network's service at a rate $l$, repopulating by doing so the $S$
compartment. Along with this birth rate, a natural death rate $m$ is considered 
for each of the three kind of nodes, while the $R$ compartment is for corrupted 
nodes in the original situation 1, as depicted in Figure~\ref{SIR3modelbirthdeath}. 
Notice that such a model is compatible with living and awaken/connected
nodes that have stopped to transfer the information in Situation 3.

To model such a scenario requires to rewrite the first line of 
Equation~\ref{modelSir2}, leading to the following system:
\begin{equation}
\label{modelSir3}
\left\{
\begin{array}{l}
\frac{dS}{dt} = l - b I S -m S\\\\
\frac{dI}{dt} = b I S - c I - m I\\\\
\frac{dR}{dt} = c I - m R.\\
\end{array}
\right.
\end{equation} 

Our objective is to know if the information will disappear from the network in the short or long term, or if it will remain available in the network, like the epidemy. The solutions in the long term, which depend on the nature of the equilibrium point, can be studied for this purpose. It leads to the phases diagram in Figure~\ref{img2123}. 

It can be shown that in this case, there is a set of parameters of the IoT for which the number of the informed nodes is strictly positive at an equilibrium point. Thus, if other solutions approach to equilibrium, the number of the informed nodes will remain strictly positive, and the information will remain in the network and become ``endemic''. We can then define, as above, a $R_0$ depending on the parameters of the system, such as the situation is endemic if and only if $R_0 >1$.

\begin{figure}[ht]
\centering
\includegraphics[scale=0.5]{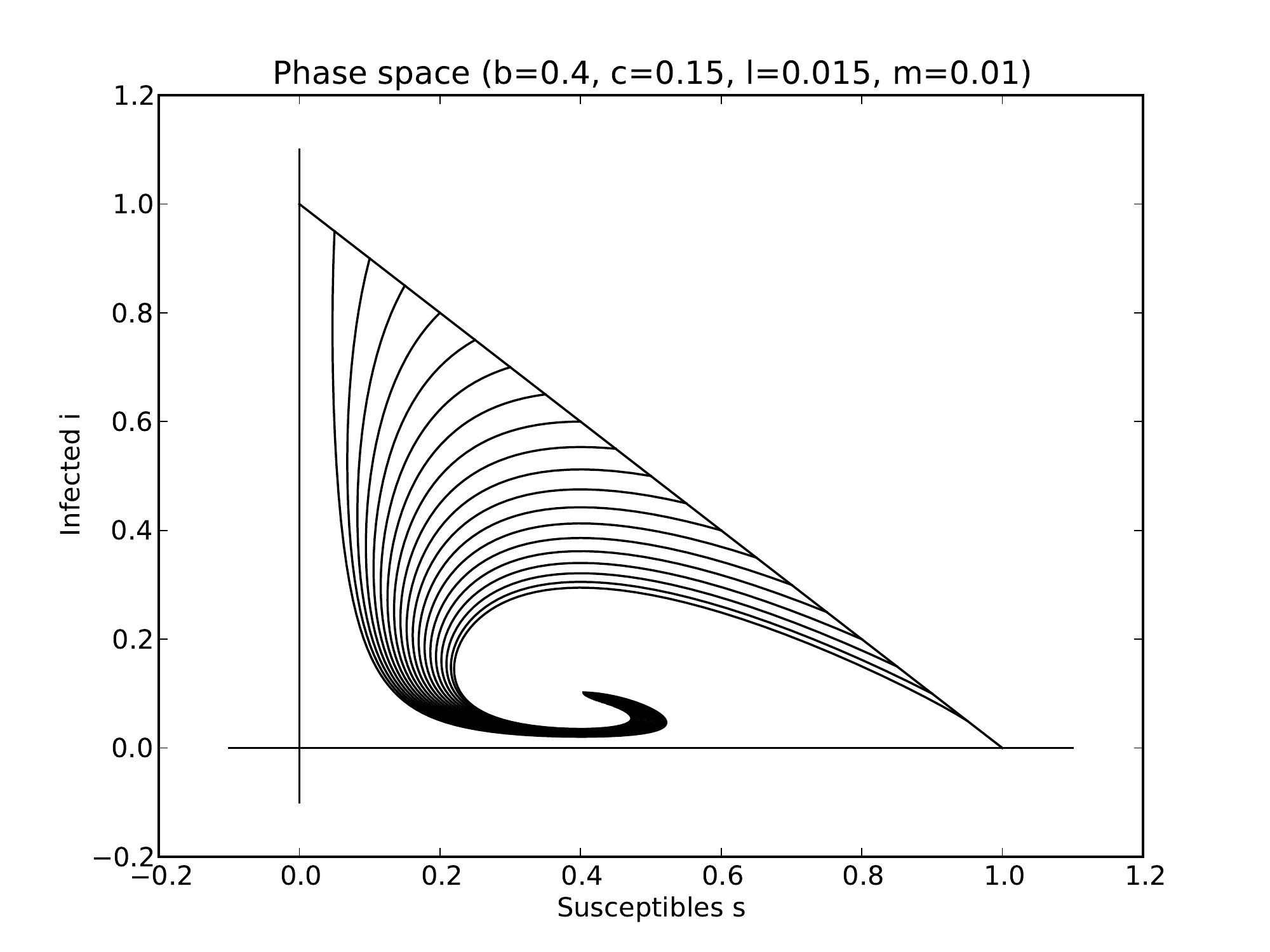}
\caption{Evolution of the fractions $s$ and $i$ of susceptible and having the datum 
nodes, SIR model with natural birth and death rates ($R_0=3.75$).}
\label{img2123}
\end{figure}

The attacker wish is to have $R_0<1$ tending to an information-free equilibrium,
whereas $R_0$ must be greater than 1 for the end user to face such an attack. If the
attacker has the opportunity to observe the network running a certain duration,
then he or she can infer the values of parameters $b, c, m$, and $l$. 
Let $N$ be the number of data transmissions by one
informed node per time unit, that is, $N=\frac{bl}{m}$. If the attacker
is able to detect and infect the informed nodes in a time $\frac{1}{c+m}$ lower than $\frac{1}{N}$,
then he or she is sure that $R_0<1$: the data will not survive in the network.
The user's interest, for its part, is to have $\frac{bl}{m}$ large and $\frac{1}{c+m}$ low, which means $R_0<1$, 
which can be achieved in the following manner:

\begin{itemize}
\item increasing the birth rate $l$,
\item increasing the lifetime of nodes to reduce $m$,
\item increasing the data transmission rate $b$, but $m$ increases when $b$ increases,
\item if possible, reducing $c$ by considering countermeasures against data removal.
\end{itemize}

\section{Numerical Simulations}

Let us now verify experimentally the efficiency of the proposal. 
A network of 10,000 interconnected nodes has been modelled using a simulator written in Python
language. It can handle 3 states for each node, which will be called here Susceptible, Infected, or Recovered. We have computed the nodes behavior
and capabilities according to what has been presented and explained previously in the theoretical part
of this article. 90\% of the nodes are initially in the $S$ compartment,
while 10\% of the network is set to Infected at time $t=0$. 
The numerical simulator is then launched during 30 time units.

\begin{figure}[h]
\subfigure[$R_0=0.001111$]{\includegraphics[scale=0.35]{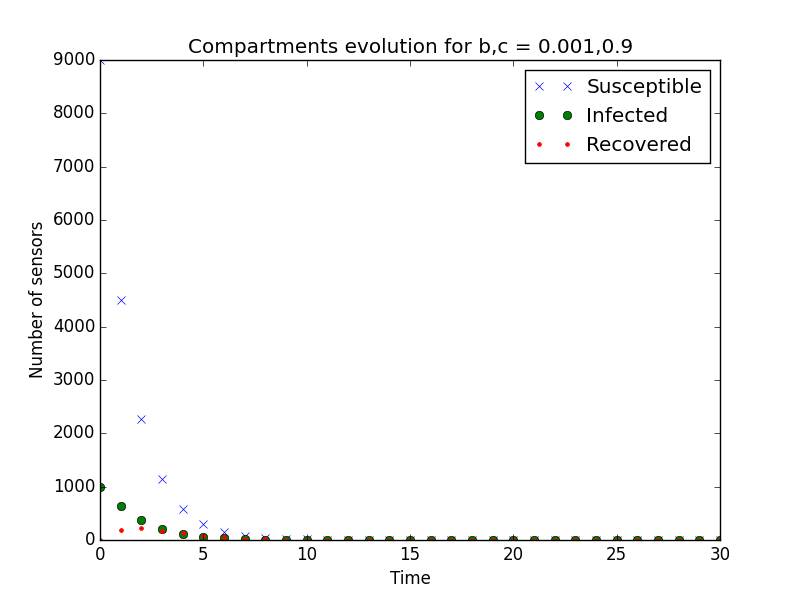}}
\subfigure[$R_0=4.9995$]{\includegraphics[scale=0.35]{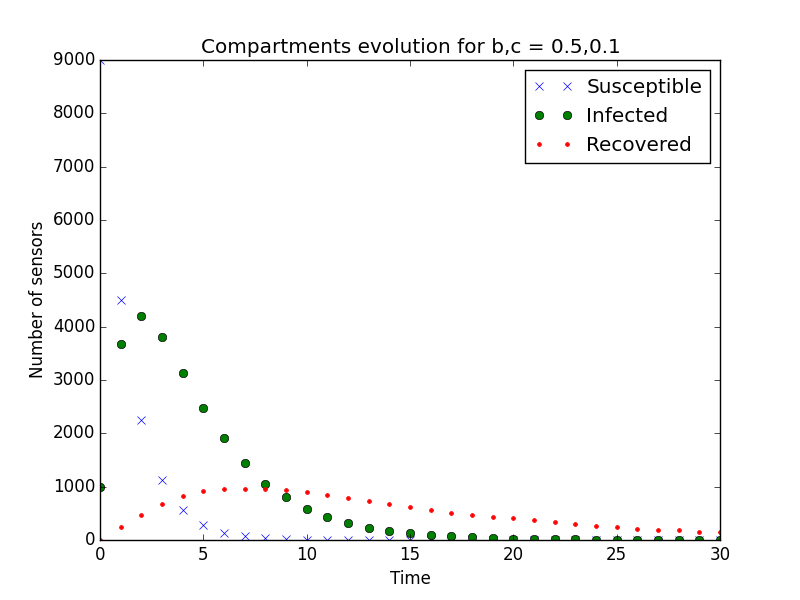}}
\caption{Simulations of Eq.~\ref{modelSir1} model and illustration of Prop.~\ref{prop1}}
\label{simul1}
\end{figure}

In all 4 experiments, nodes from the $S$ compartment can be infected following a probability $bI$,
which depends on the number of infected nodes, while the nodes in the $I$ compartment 
move to the $R$ compartment according to a $c$ probability. $b$ and $c$ are respectively
set to $(0.001, 0.9), (0.5, 0.1), (0.2, 0.15),$ and $(0.23, 0.01)$ in the $4$ experiments,
to see the infection rates effects on the network evolution.
In the first two experiments we simply consider the SIR model described
in Figure~\ref{SIRmodel}, while in the two last simulations we respectively add
one and two death rates $m$ and $m'$, corresponding to the probability for a
given node to have emptied its battery.

The objective of the first two simulations
is to illustrate experimentally the effects of the constant $R_0$ on the evolution of the
number of infected nodes, what has proven in Proposition~\ref{prop1}. It can be seen
in Figure~\ref{simul1} that the $I$ compartment has a maximum for $R_0>1$, while
it is not the case if $R_0<1$. Figure~\ref{simul2}, for its part,
illustrates the evolution of the network when considering death rates 
in the 3 variations of SIR model described in Section~\ref{VSIR}.
We can see that the $I$ 
compartment is never empty, leading to a data survivability in this IoT, while
phase spaces and compartment evolution theoretically deduced in a previous 
section can be experimentally obtained too.

\begin{figure}[h]
\subfigure[Situation 2]{\includegraphics[scale=0.35]{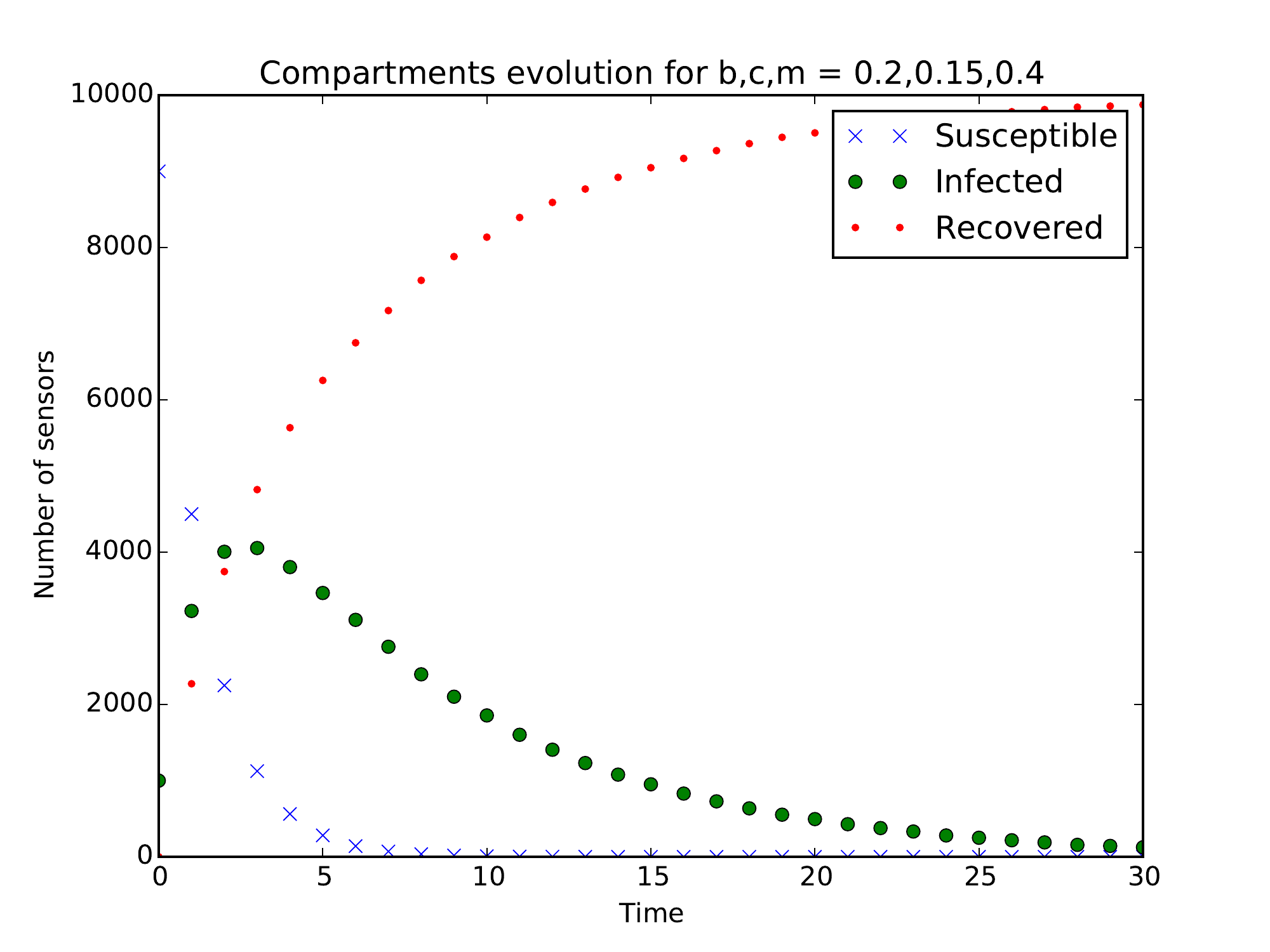}}
\subfigure[Situations 1 and 3]{\includegraphics[scale=0.35]{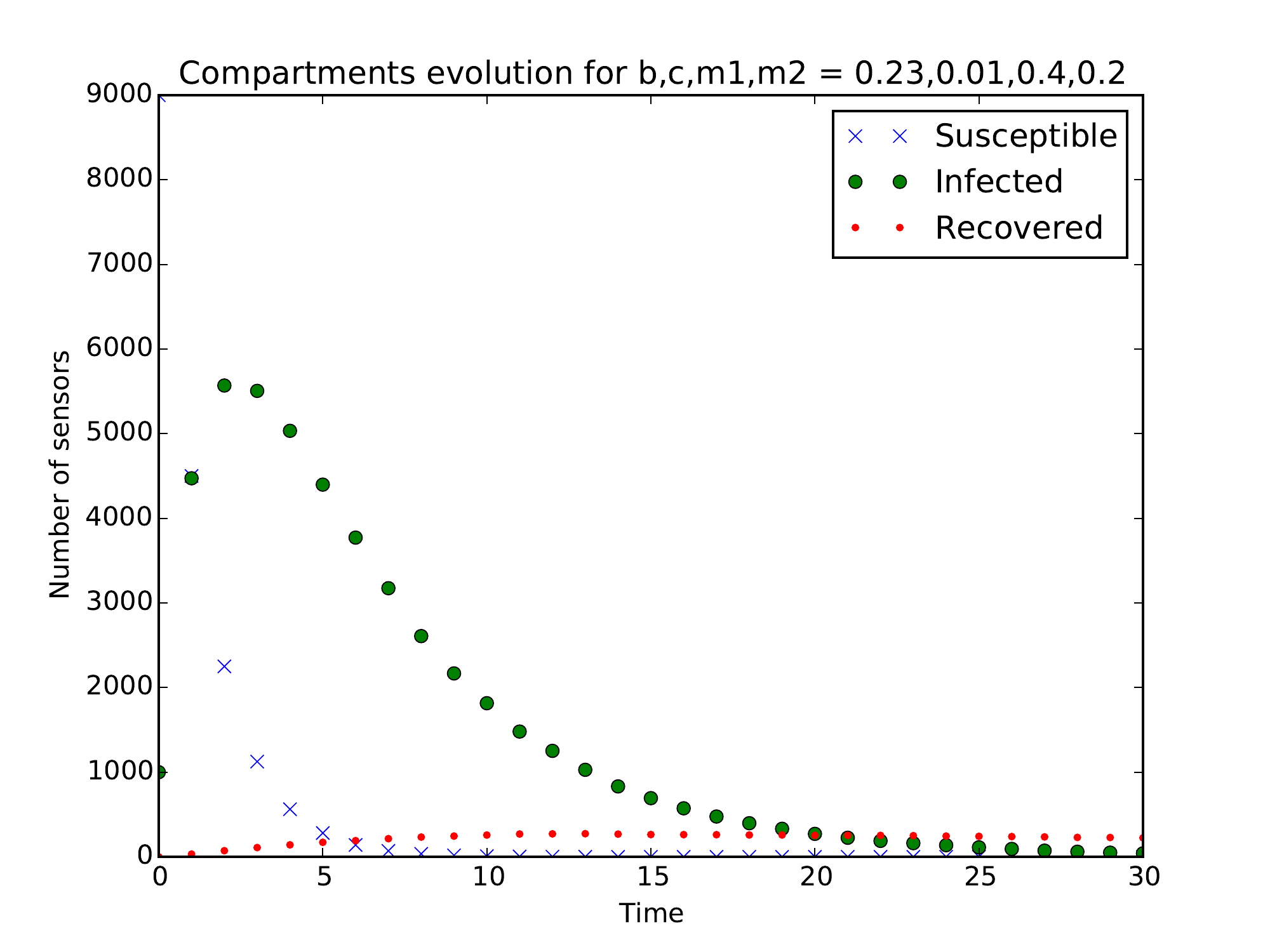}}
\caption{Compartment evolution in the 3 variations of SIR model described in Section~\ref{VSIR}.}
\label{simul2}
\end{figure}

For the sake of comparison, we have regarded which results can be obtained with a situation where 
each informed node send the data to its $k$ neighbors as soon as it receives it, 
for various $k$. The trade off in that situation is as follows: increasing the $k$ value
leads to a better data survivability, but the network's lifetime is consequentially reduced
accordingly. This obvious fact is illustrated in Figure~\ref{simulk}, in which the number of 
informed nodes increases more rapidly with a larger $k$, while the death rate is 
increasing too. Remark that one objective of the theoretical study presented previously 
were indeed to find the best $k$ in such a situation.

\begin{figure}[h]
\subfigure[$k=1$]{\includegraphics[scale=0.35]{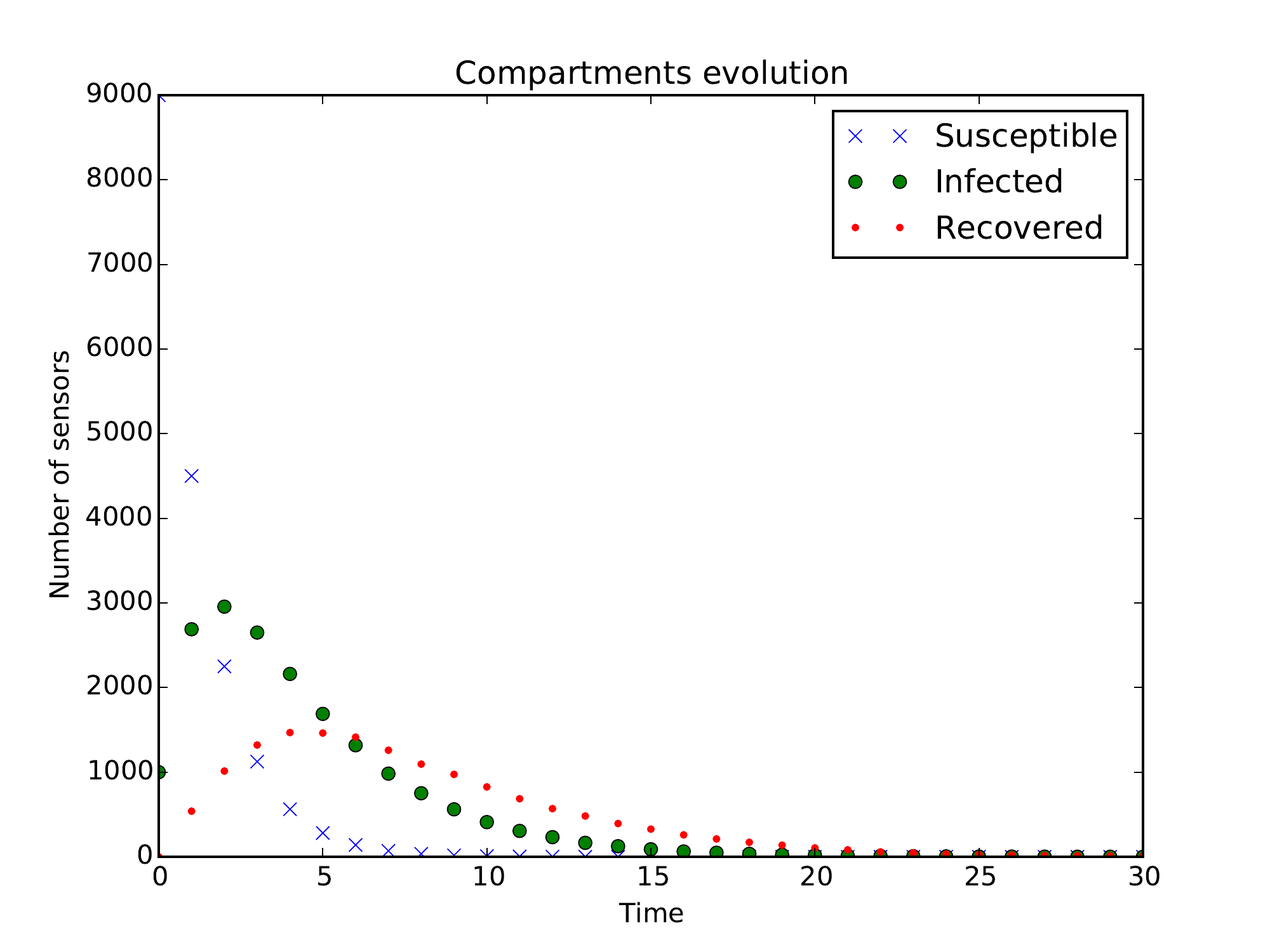}}
\subfigure[$k=4$]{\includegraphics[scale=0.35]{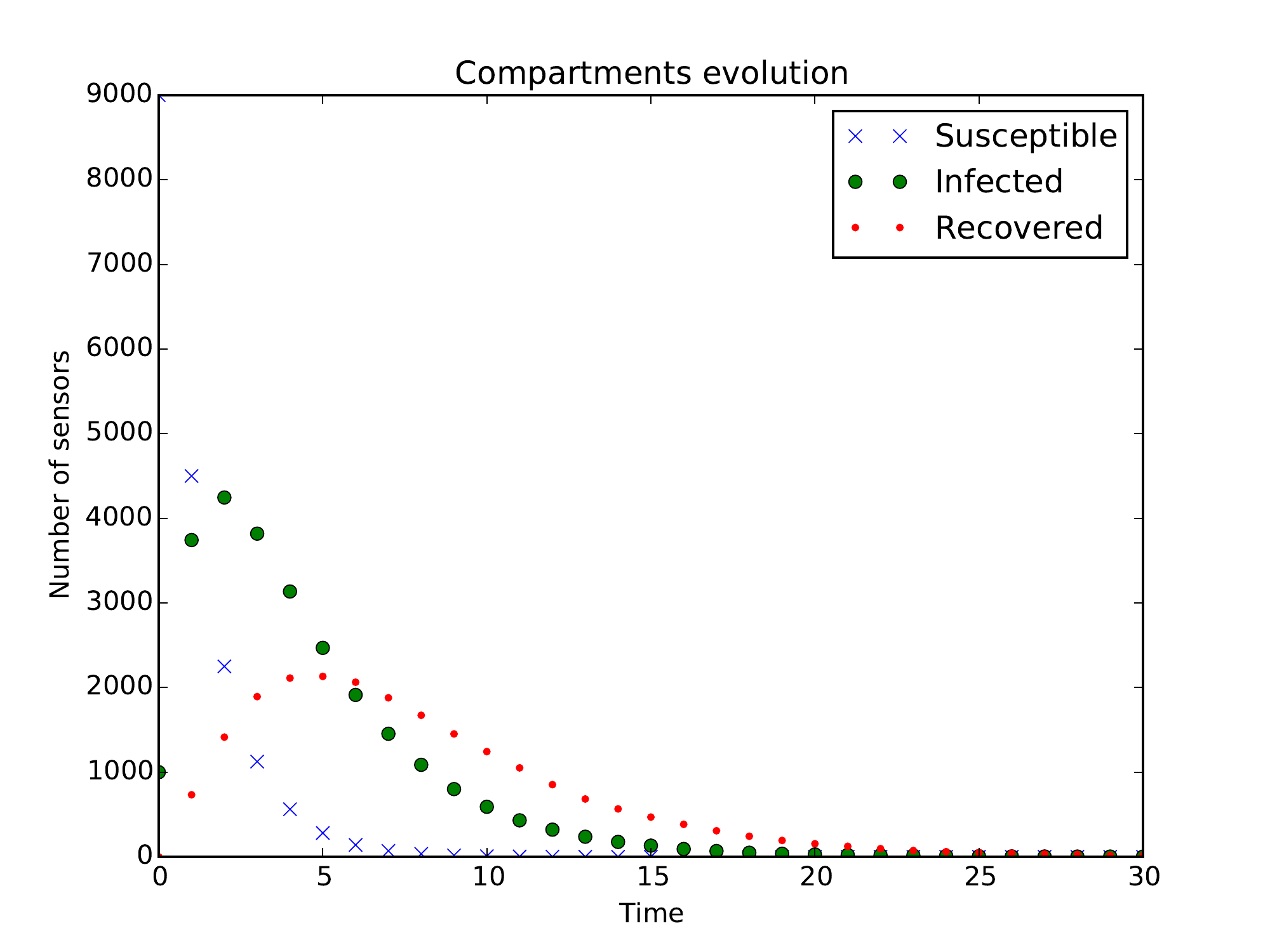}}
\caption{Compartment evolution when each informed node send the data to its $k$ neighbors.}
\label{simulk}
\end{figure}

We investigated too a random approach, in which each node picks randomly both the number of 
closest neighbors it must transfer the information and in how much iterations it must achieve
the transfer. Obtained results are presented in Figure~\ref{figRandom}. In that situation we observed that,
even though the data survivability and network's lifetime are locally heterogeneous, behaviors are
averaged on the whole network, leading to a global result that does not outperform what we obtained
when computing the optimal $k$ in a previous experiment. Finally, in our last experiment, the 
number of data transmission became inversely proportional to the battery level of the associated sensor.
As can be seen in Figure~\ref{figAuto}, data survivability is improved during the first iterations, but this 
latter decreases dramatically over time. This is not surprising, as the data transmission decreases
accordingly.

\begin{figure}[h]
\subfigure[Random approach\label{figRandom}]{\includegraphics[scale=0.35]{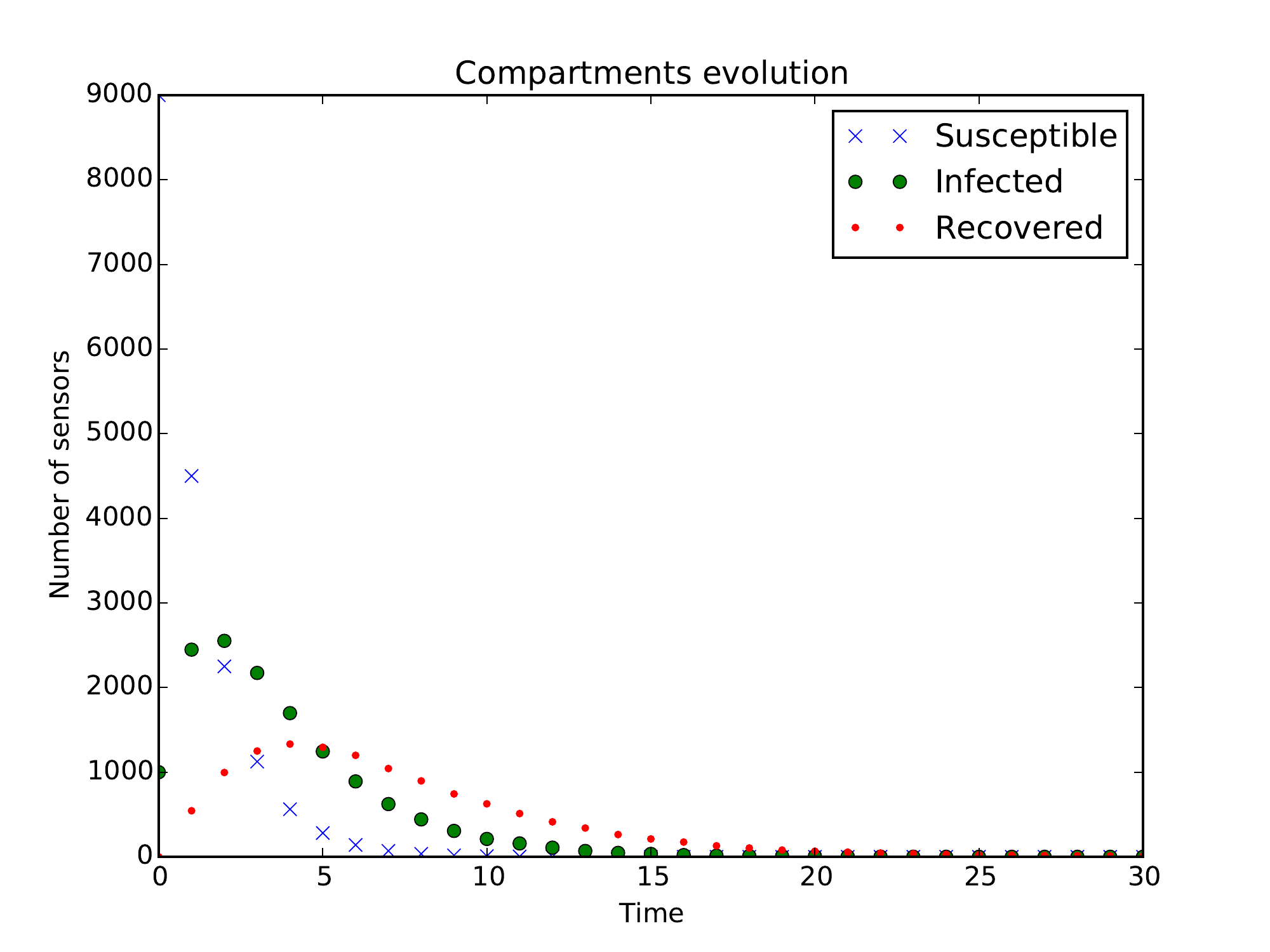}}
\subfigure[Transmission inversely proportional to the battery\label{figAuto}]{\includegraphics[scale=0.35]{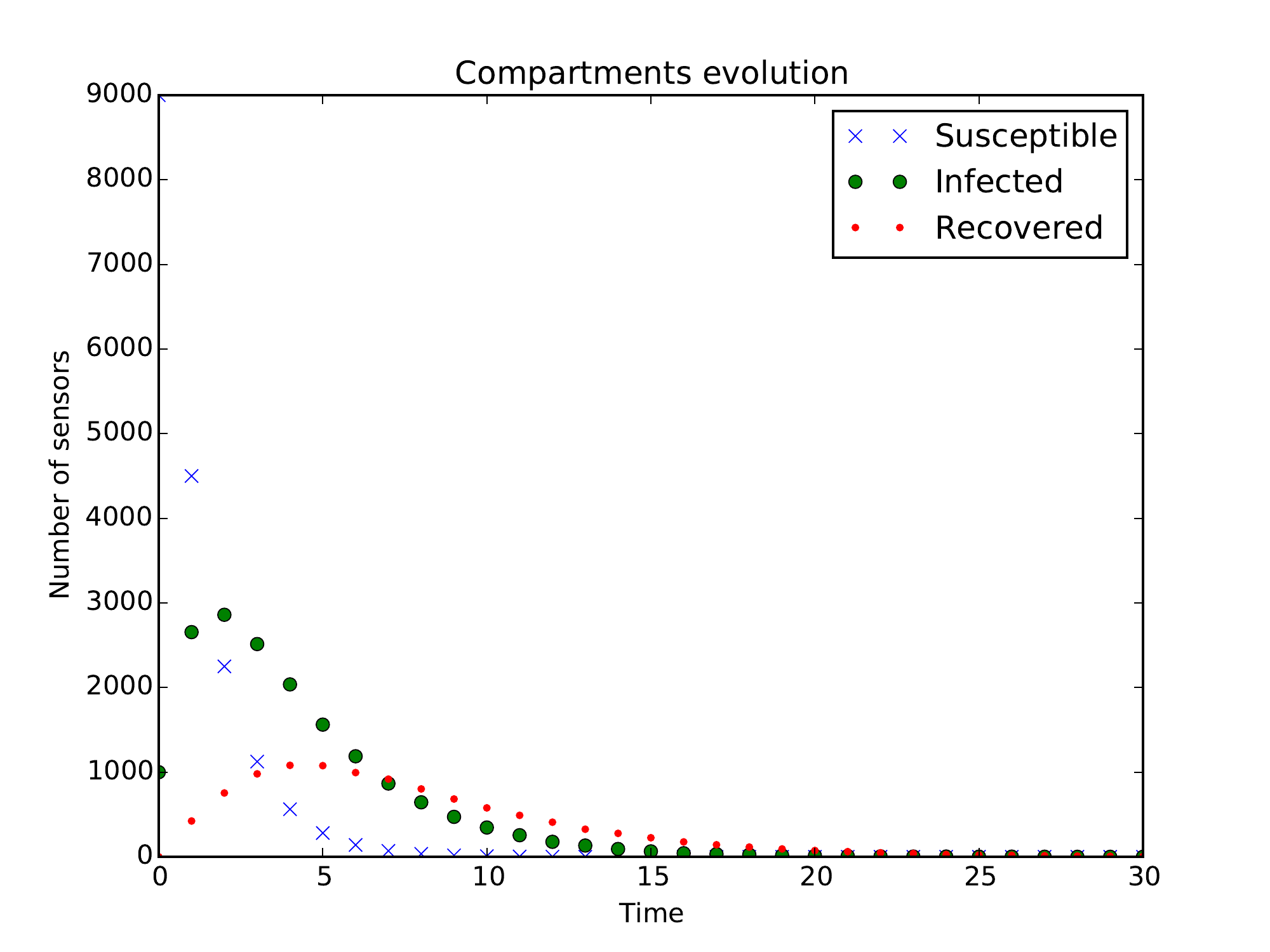}}
\caption{Compartment evolution for heterogeneous transmission rates.}
\end{figure}

\section{Conclusion}
\label{CONC}

Data survivability is a main challenge of the internet of things that can not be ignored. Therefore, collaborating and the transmitting crucial information between nodes are to be realized in order to maximize the amount of monitoring-related data that can survive. This paper presented an efficient technique that uses epidemic domain models in the context of data survival in the IoT. We modelled the collaboration and transmission between nodes via a SIR (Susceptible - Infected - Recovered) model that can ensure the survivability of the datum in presence of different types of attacks. We showed that our method is well adapted to IoT applications based on wireless sensor networks with energy and resource constraints. On the other hand, it takes into account the dynamic network topology which is a non-negligible constraint. The introduced model which we prove its viability for IoT paves the way for further investigations in the IoT domain. For instance, assessing how these
results are influenced by nodes or things mobility. Furthermore, in order to evaluate the efficiency of the proposed technique, real experiments are
planned for the future.

\section*{Acknowledgements}
This work is partially funded by the Labex ACTION program (contract ANR-11-LABX-01-01).

\bibliographystyle{plain}
\bibliography{biblio}
\end{document}